\def\draft{0}  
\documentclass[11pt,tikz]{article}
\usepackage[letterpaper,hmargin=1in,vmargin=1in]{geometry}
\usepackage{ifthen}
\newcommand{\talkingPoint}[1]{\ifthenelse{\equal{\draft}{1}}{{\color{blue}{---#1---}}}{}}
\newcommand{\KM}[1]{\ifthenelse{\equal{\draft}{1}}{{\color{cyan}{---KM:#1---}}}{}}
\newcommand{\XD}[1]{\ifthenelse{\equal{\draft}{1}}{{\color{red}{---XD:#1---}}}{}}
\newcommand{\YSH}[1]{\ifthenelse{\equal{\draft}{1}}{{\color{purple}{---YS:#1---}}}{#1}}

\usepackage{floatrow}
\usepackage{tikz}
\usetikzlibrary{automata}

\usepackage{enumitem}
\usepackage{geometry,color}                
\usepackage{fullpage}
\usepackage{graphicx}
\usepackage{amssymb}
\usepackage{epstopdf}
\usepackage{amsfonts,latexsym,graphics}
\usepackage{epsfig}
\usepackage{amssymb}
\usepackage{amsmath}
\usepackage{amsfonts}
\usepackage{algorithm}
\usepackage{algorithmic}
\usepackage{url}
\usepackage{rotating}
\usepackage{hyperref}
\usepackage{scalefnt}
\usetikzlibrary{automata,decorations.pathmorphing}

\ifnum\draft=0
\newcommand{\Rnote}[1]{{[\bf Rafael's Note: #1]}}
\newcommand{\Knote}[1]{{[\bf Kai-Min's Note: #1]}}
\newcommand{\Wnote}[1]{{[\bf Xiaodi's Note: #1]}}
\newcommand{\Snote}[1]{{[\bf Yaoyun\rq{} Note: #1]}}
\else
\newcommand{\Rnote}[1]{{}}
\newcommand{\Knote}[1]{{}}
\newcommand{\Wnote}[1]{{}}
\newcommand{\Snote}[1]{{}}
\fi

\newtheorem{theorem}{Theorem}
\newtheorem{definition}[theorem]{Definition}

\newtheorem{lemma}[theorem]{Lemma}

\newtheorem{proposition}[theorem]{Proposition}

\newtheorem{fact}[theorem]{Fact}

\newtheorem{remk}[theorem]{Remark}

\newenvironment{proof}{\noindent{\bf Proof. }}{\qed}

\def\FullBox{\hbox{\vrule width 8pt height 8pt depth 0pt}}

\def\qed{\ifmmode\qquad\FullBox\else{\unskip\nobreak\hfil
\penalty50\hskip1em\null\nobreak\hfil\FullBox
\parfillskip=0pt\finalhyphendemerits=0\endgraf}\fi}

\def\qedsketch{\ifmmode\Box\else{\unskip\nobreak\hfil
\penalty50\hskip1em\null\nobreak\hfil$\Box$
\parfillskip=0pt\finalhyphendemerits=0\endgraf}\fi}

\newenvironment{proofof}[1]{\begin{trivlist} \item {\bf Proof #1:~~}}
  {\qed\end{trivlist}}

\newcommand{\poly}{{\mathrm{poly}}}

\newcommand{\zo}{\{0,1\}}

\newcommand{\E}{\mathop{\mathrm E}\displaylimits}

\newcommand{\remove}[1]{}

\newcommand{\class}[1]{{\mathbf{#1}}} 

\renewcommand{\L}{\class{L}}

\newcommand{\eps}{\varepsilon}

\newcommand{\Zsystem}{Z}

\newcommand{\DecisionSystem}{A}

\newcommand{\AcceptEvent}{{A}}
\newcommand{\RejectEvent}{{\bar A}}

\def\01{\{0,1\}}
\def\eps{\epsilon}

\newcommand{\Prob}{{\mathbf{Pr}}}
\newcommand{\tinyspace}{\mspace{1mu}}
\newcommand{\microspace}{\mspace{0.5mu}}

\newcommand{\norm}[1]{\left\lVert\tinyspace#1\tinyspace\right\rVert}

\newcommand{\defeq}{\stackrel{\smash{\text{\tiny def}}}{=}}
\newcommand{\tr}{\operatorname{tr}}

\newcommand{\ip}[2]{\left\langle #1 , #2\right\rangle}

\def\({\left(}
\def\){\right)}
\def\I{\mathsf{id}}

\newcommand{\fid}{\operatorname{F}}
\newcommand{\setft}[1]{\mathrm{#1}}
\newcommand{\lin}[1]{\setft{L}\left(#1\right)}
\newcommand{\density}[1]{\setft{Dens}\left(#1\right)}

\newcommand{\ot}{\otimes}

\def\complex{\mathbb{C}}

\def\<{\langle}
\def\>{\rangle}
\def \lket {\left|}
\def \rket {\right\rangle}

\def \rbra {\right|}
\newcommand{\ket}[1]{\lket\microspace #1 \microspace\rket}

\newcommand{\ketbra}[1]{\lket\microspace #1 \rangle\langle #1 \microspace\rbra}

\def\X{\mathcal{X}}
\def\Y{\mathcal{Y}}

\def\A{\mathcal{A}}
\def\B{\mathcal{B}}

\def\E{\mathcal{E}}

\def\K{\mathcal{K}}

\newcommand{\trnorm}[1]{\norm{#1}_{\tr}}

\newcommand{\uniform}[1]{\mathcal{U}_{#1}}

\def\defeq{\stackrel{\small \textrm{def}}{=}}

\newcommand{\commentout}[1]{}
\numberwithin{theorem}{section}
\numberwithin{equation}{section}

\newenvironment{step}
  {
    \begin{enumerate}

  }
  {\end{enumerate}}

\newenvironment{protocol*}[1]
  {
    \begin{center}
      \hrulefill\\
      \textbf{#1}
  }
  {
    \vspace{-1\baselineskip}
    \hrulefill
    \end{center}
  }

\newcommand{\Ext}{{\mathrm{Ext}}}
\newcommand{\Hmin}{\mathrm{H}_{\infty}}
\newcommand{\vecD}{{D}}
\newcommand{\vecS}{{\bf{S}}}
\newcommand{\Piamp}{{\PExt}}
\newcommand{\Picert}{{\PExt_{\mathrm{seed}}}}
\newcommand{\tamp}{{t_{\sf PRE}}}
\newcommand{\tcert}{{t_{\sf seed}}}

\newcommand{\Adv}{E}
\newcommand{\Abs}{X}
\newcommand{\PS}{\mathcal{S}}

\newcommand{\QPS}{\mathcal{S}}

\newcommand{\Alg}{A}

\title{Physical Randomness Extractors: Generating Random Numbers with Minimal Assumptions}

\author{%
 Kai-Min Chung\footnote{Institute of Information Science, Academia Sinica, Taiwan.}  $\qquad$
 Yaoyun Shi\footnote{
Department of Electrical Engineering and Computer Science,
  University of Michigan, Ann Arbor, MI 48103, USA.}
 $\qquad$ Xiaodi Wu\footnote{Center for Theoretical Physics,
Massachusetts Institute of Technology, Cambridge, MA 02139, USA. Most of the research was conducted when the author was a student at the University of Michigan, Ann Arbor.}
}
\date{March 9, 2015}                                           

\begin{document}

\begin{titlepage}
\thispagestyle{empty}
\maketitle
\thispagestyle{empty}
\begin{abstract}
\noindent
How to generate provably true randomness with minimal assumptions? This question is important not only for the efficiency and the
security of information processing, but also for understanding how extremely 
unpredictable events are possible in Nature. All current solutions
require special structures in the initial source of randomness,
or a certain independence relation among two or more sources.
Both types of assumptions are impossible to test and difficult to guarantee
in practice.
Here we show how this fundamental limit can be circumvented by
extractors that base security on the validity of physical laws and extract randomness
from {\em untrusted} quantum devices.
In conjunction with the recent work of Miller and
Shi (arXiv:1402:0489), our {\em physical randomness extractor} uses just a
single and general weak source, produces an arbitrarily long
and near-uniform output, with a close-to-optimal error,
secure against all-powerful quantum adversaries, and tolerating a constant level of implementation
imprecision. The source necessarily needs to be unpredictable to the devices,
but otherwise can even be known to the adversary.

Our central technical contribution, the Equivalence Lemma, provides a general
principle for proving composition security of untrusted-device protocols. It implies that {\em
unbounded} randomness expansion can be achieved simply by cross-feeding {\em any} two expansion protocols.
In particular, such an unbounded expansion can be made robust, which is known for the first time.
Another significant implication is, it enables the secure randomness generation and key distribution
using {\em public} randomness, such as that broadcast
by NIST's Randomness Beacon. Our protocol also provides a method for refuting local hidden variable theories 
under a weak assumption on the available randomness for choosing the measurement settings.

\end{abstract}
\end{titlepage}

\pagebreak
\thispagestyle{empty}
\noindent{\bf Version Differences.}
This draft differs substantially from 
both the first and the second versions of our arXiv posting (arXiv:1402.4797).
\begin{itemize}
\item V1 of arXiv:1402.4797 is our QIP 2014 submission (it was accepted and presented in a joint-plenary presentation.)
\item V2 of arXiv:1402.4797 introduces and formally defines the notion of physical randomness extractors. The Master Protocol
is also changed slightly so that for arbitrarily small min-entropy sources the protocol is robust. The new analysis requires some thought.
\item This current version is a substantial re-writing of V2. The main technical new material is in the formal definition of
physical extractors. The proofs are correspondingly changed. In particular, an abstract notion of error model is added, and the robustness
claim in V2 is now rigorous.  On the other hand, this version focuses on quantum-security. 
We assume the completeness of quantum mechanics, which means that
the adversary cannot obtain any information other through quantum operations.
A planned update of this work will address non-signaling security.

\end{itemize}

\pagebreak
\newcommand{\abstracts}{abstract }
\newcommand{\physical}{untrusted }
\newcommand{\Phy}{\vecD}
\newcommand{\vecPhy}{\Phy}
\newcommand{\PRE}{PRE }
\newcommand{\PExt}{\mathrm{PExt}}
\def\Proc{P}
\newcommand{\Procamp}{{\Proc_{\sf amp}}}
\newcommand{\Proccert}{{\Proc}}
\newcommand{\ndevice}{t}

\setcounter{page}{1}
\section{Motivations}
\talkingPoint{importance of randomness}
Randomness is a vital resource
for modern day information processing. The wide range of its applications 
include cryptography, fast randomized algorithms,
accurate physical simulations, and fair gambling.
\talkingPoint{the current solutions require assumptions}
In practice, randomness is generated through a ``random number generator'' (RNG),
such as Intel's on-chip hardware generator \texttt{RdRand}
and Linux's software generator \texttt{/dev/random}.
Since it is impossible to test if the output of a RNG is uniformly
distributed or fixed,~\cite{convex}
one relies on the mathematical properties
of the RNG to ensure the output quality under a set of assumptions
that are hopefully true in reality.
For example, Linux's RNG critically requires being seeded
with a large amount of initial entropy 
and the unproven assumption that no adversary is computationally powerful enough to
differentiate the output from uniform.

\talkingPoint{those assumptions fail them}
Those assumptions, however, have been repeatedly shown
to cause failures of practical RNGs~(see, e.g., \cite{GuttermanPR06,Ristenpart:2010,Heninger:2012,Lenstra+,NYT:Snowden}).
Such vulnerabilities of RNGs directly threaten the very foundation of digital security,
and risk invalid conclusions drawn from computations assuming
true randomness. Thus when security is of paramount importance, it is 
highly desirable to use RNGs that are secure
under a minimal set of assumptions.

\talkingPoint{classical theory and the independence barrier}
The classical theory for this objective is that 
of {\em randomness extractors}~\cite{Vadhan07}.
In this theory, an extractor is a deterministic algorithm that transforms several sources
of weak randomness into near-perfect randomness. The amount of randomness in 
a weak source is quantified by {\em min-entropy},
or conditional quantum min-entropy when the adversary is quantum.
More precisely, an $(n, k)$ source is an $n$-bit binary string with (conditional quantum)
min-entropy $k$, which means that the best chance
for an adversary to guess the source correctly is $\le 2^{-k}$~\cite{Zuckerman:1990, Renner:thesis,RennerWW:2006}.
A fundamental limit known in this theory is that 
randomness extraction
is possible only when {\em two or more independent} sources are available.
In particular, deterministic extraction,
i.e., using just one source, is known to be impossible to produce even $1$ (near-prefect) random bit~\cite{SanthaV84}.
Since independence is impossible to check~\cite{independence}
and difficult to guarantee in practice, 
the classical theory of randomness extractors inevitably relies on {\em assuming} {independence}.

\talkingPoint{quantum solution and the trust problem}
Quantum mechanics has perfect randomness in its postulate, thus appears to 
provide a simple solution to the problem~\cite{quantum_measurements}.
Indeed, commercial products are already available (e.g., the Quantis generators of ID Quantique).
However, users must trust the quantum devices in use for security. This is a strong assumption undesirable
in certain circumstances for the following reasons.
First, as classical beings, we can only directly verify classical information, thus cannot 
directly verify the inner-workings of quantum devices. Second, we may not want to trust the manufacturers or the certifying government agencies.
Finally, even if the manufacturers are truthful, the devices themselves may not work properly due to
technological limitations. No method is currently known for reliably implementing quantum devices in a large scale.

\talkingPoint{a new solution space and the limitations of known works}
Recent works have shown that one can still leverage
the quantum power for generating randomness even when the underlying quantum
devices may be imperfect, or even 
malicious~\cite{col:2006, ColbeckK:2011, Pinorio, VV12, MS, CY, ColbeckR:2012, Gallego:2012, Brandao:amplification}.
However, all those protocols also crucially rely on a certain form of independence.
More specifically, {\em randomness expansion} protocols require that the source
is globally uniform (thus independent
from the devices)~\cite{col:2006,ColbeckK:2011, Pinorio, VV12, MS, CY}.
{\em Randomness amplification} protocols~\cite{ColbeckR:2012, Gallego:2012, Brandao:amplification}
initiated by Colbeck and Renner~\cite{ColbeckR:2012}
require that the source, when conditioned on the adversary's side information,
is a highly random and highly structured SV-source~\cite{SV-source}.
In~\cite{ColbeckR:2012}, the source in addition needs to satisfy certain causal relation among its blocks,
while in~\cite{Gallego:2012, Brandao:amplification}, conditioned on the adversary's information, the source and
the devices are assumed to be independent.~\footnote{The conditional independence condition is stated explicitly in the work Gallego et al.~\cite{Gallego:2012} (Supplementary Note 2) and Brandao et al.~\cite{Brandao:amplification} (Section II.B).
The causal relation assumed by Colbeck and Renner~\cite{ColbeckR:2012} is illustrated in their Fig.~2 and stated below it.}


\section{Our contributions}
\paragraph{Physical Randomness Extractors: a model for extracting randomness without independence assumptions.}
To circumvent those fundamental limits and to minimize necessary assumptions, 
here we formulate a framework of extracting randomness from {\em untrusted}
quantum devices in the quantum mechanical world, shown in Fig.~\ref{fig:phy_ext}. 
This framework of {\em Physical Randomness Extractors (PREs)} allows general and rigorous discussions of extracting 
randomness when the devices and the adversary are both bound by physical laws.
This reliance on physical theories for security is a fundamental departure from the classical theory of randomness
extraction. Since all cryptographic protocols will eventually be deployed in the physical
world, no additional effort needs to be made to enforce the assumptions on the correctness
of physical laws (as Nature automatically ensures that)~\cite{general}.

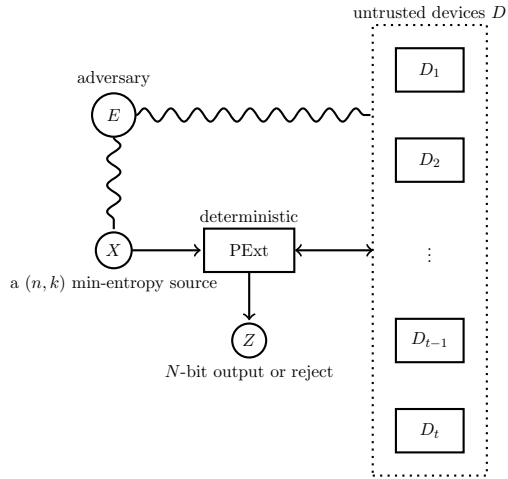
\begin{figure}[h]
\begin{center}
\begin{tikzpicture}[shorten >=1pt,node distance=2cm,thick,scale=.6, every node/.style={scale=.6}]
\node [circle,draw, label={below: a $(n,k)$ min-entropy source}] (start_n) {$X$};
\node [state, circle,label={above:adversary}] (Eve) [above of = start_n,yshift=1cm] {$E$};
\node [state, rectangle, minimum width=2cm] (ext) [right of =start_n, xshift=1cm,label=above:\textrm{deterministic}] {$\PExt$};
\node (ext_node) [right of = ext] {};
\node (dots) [right of = ext_node] {$\vdots$};
\node [state, rectangle, minimum width=1.5cm] (P_1) [above of = dots] {$\Phy_2$};
\node [state, rectangle, minimum width=1.5cm] (P_0) [above of = P_1] {$\Phy_1$};
\node [state, rectangle, minimum width=1.5cm] (P_s) [below of = dots] {$\Phy_{\ndevice-1}$};
\node [state, rectangle, minimum width=1.5cm] (P_l) [below of = P_s] {$\Phy_{\ndevice}$};
\node [rectangle, minimum size=10cm, dotted]  (phy) [right of = ext_node, label=above:\textrm{untrusted devices $\vecPhy$}]{};
\node [circle, draw](out) [below of = ext, label=below:\textrm{$N$-bit output or reject}]{$Z$};

\draw [dotted] ([xshift=-.5cm, yshift=.5cm]P_0.north west) rectangle ([xshift=.5cm,yshift=-.5cm]P_l.south east);
\path[<->]
(ext) edge node {} ([xshift=-1cm] dots.west);
\path[->]
(start_n) edge node {} (ext)
 (ext) edge node {} (out);
 \draw[decorate,decoration=snake] (Eve) -- (start_n);
 \draw[decorate,decoration=snake] (Eve) -- ([xshift=5.25cm]Eve.east);

\end{tikzpicture}
\caption{\small Physical Randomness Extractor (PRE). A PRE is a deterministic algorithm $\PExt$ that takes a classical source $X$ as the input,
interacts with a set of untrusted quantum devices $D$, and finally either aborts (aka rejects) or outputs a binary string $Z$. Each device
is used through its classical input-output interface but its inner-working is unknown (and could be malicious). The Adversary $E$ is quantum and all-powerful,
may be in an unknown quantum correlation with the devices, and together with the devices may hold a certain amount of side information about $X$. After the protocol starts, no communication is allowed among the Adversary and the devices.
The error of $\PExt$ upper-bounds
both the probability of accepting an undesirable output (soundness error) or that of rejecting an honest implementation (completeness error). If $X$ is globally uniformly random, $\PExt$ is said to be {\em seeded}; otherwise, it is {\em seedless}.
$\PExt$ is {\em robust} if an honest implementation can deviate from an ideal implementation by a constant amount. See Section~\ref{sec:formal_model} for the formal definitions. 
}
\end{center}
\label{fig:phy_ext} 
\end{figure}

Our framework is built upon the above-mentioned two lines of research on randomness expansion
and randomness amplification. It in particular includes the quantum restriction of those models~\cite{nonsignaling}.
As special cases.
{Randomness expansion}
is precisely {\em seeded} PRE-extraction, where the seed is uniform globally. {Randomness amplification}, when restricted
to the quantum world,
can be seen as {\em seedless} (i.e., the classical source is not uniform) PRE-extraction of a single bit with a restricted source.
Our framework explicitly quantifies the various relevant resources.
This allows richer analyses and comparisons of protocols, and raises new questions for optimizing the performance parameters
and investigating their inherent tradeoffs. For example, the {\em extraction rate} introduced, 
i.e., the ratio of the output length and the total length of the device output, is a natural measure
for the efficiency of a PRE (See Section~\ref{sec:formal_model} for more details.) We discuss several fundamental
open problem in Section~\ref{sec:open}.

We point out that it'd be more appropriate to consider the untrusted devices as the source of the output randomness,
while the classical source is used to prevent cheating. This intuition is supported by the strong quantitative relations
linking the min-entropy of the source to the error parameter, and the number of device usages to the output length.
A useful comparison of a PRE with classical strong extractor is to consider the weak source of a PRE corresponds to the
seed for the latter, while the devices correspond to the weak source. Under this correspondence there is a fundamental
difference between those two models regarding the correlation vs independence of the corresponding two sources.

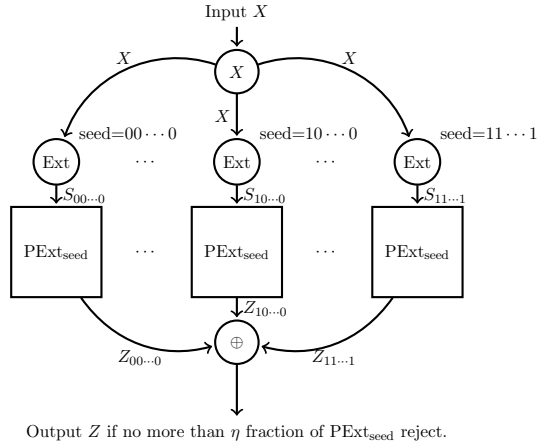
\begin{figure}[h]
\begin{center}
\begin{tikzpicture}[shorten >=1pt,node distance=2cm,thick,scale=.6, every node/.style={scale=.6}]
\node [state, initial above, initial text=Input $X$] (source) {$X$};
\node [state,minimum size=1cm ] (ext_mid) [below of = source, label=above right:\textrm{seed=$10\cdots0$}] {$\textrm{Ext}$};
\node (dots_1) [left of = ext_mid] {$\cdots$};
\node (dots_2) [right of = ext_mid] {$\cdots$};
\node [state,minimum size=1cm ] (ext1) [left of = dots_1, label=above right:\textrm{seed=$00\cdots0$}] {$\textrm{Ext}$};
\node [state,minimum size=1cm ] (ext2) [right of = dots_2, label=above right:\textrm{seed=$11\cdots1$}] {$\textrm{Ext}$};
\node [rectangle,minimum size=2cm,draw] (SRC_mid) [below of = ext_mid] {$\PExt_{\mathrm{seed}}$};
\node (dots_1) [left of = SRC_mid] {$\cdots$};
\node (dots_2) [right of = SRC_mid] {$\cdots$};

\node [rectangle, minimum size=2cm, draw] (SRC1) [below of = ext1] {$\PExt_{\mathrm{seed}}$};
\node [rectangle, minimum size=2cm, draw] (SRC2) [below of = ext2] {$\PExt_{\mathrm{seed}}$};

\node [state] (XOR) [below of = SRC_mid] {$\oplus$};
\node (out) [below of = XOR] {Output $Z$ if no more than $\eta$ fraction of $\Picert$ reject.};
\path[->]
(XOR) edge node [right]{} (out)
(source)  edge [bend right=45] node [above] {$X$} (ext1)
	  edge [bend left=45] node [above] {$X$} (ext2)
	  edge  node [left] {$X$} (ext_mid)
(ext1)  edge [right]node  {$S_{00\cdots0}$} (SRC1)
(ext2)  edge [right]node  {$S_{11\cdots1}$} (SRC2)
(ext_mid)  edge [right]node  {$S_{10\cdots0}$} (SRC_mid)
(SRC1)   edge [bend right=35, below] node  {$Z_{00\cdots0}$} (XOR)
(SRC2)  edge [bend left=35,below]  node  {$Z_{11\cdots1}$} (XOR)
(SRC_mid)  edge [right]node  {$Z_{10\cdots0}$} (XOR);
\end{tikzpicture}
\caption{\small Our Physical Randomness Extractor $\Piamp$ with parameters $\Ext$, $\Picert$, and $\eta$. $\Ext$ is a quantum-proof
strong extractor~\cite{extractor_def} and $\Picert$ a
seeded-PRE whose input length equals the output length of $\Ext$.
For each distinct seed value $i$ of $\Ext$, run an instance of $\Ext$ with that seed value and $X$ as the source.
Use the output $S_i$ as the input to a separate instance of $\Picert$. Output the XOR of the $Z_i$'s,
or abort if $\ge\eta$ fraction of $\Picert$ aborted.
}
%
\end{center}
\label{fig:prot1}
\end{figure}

\paragraph{An explicit construction.} We further construct the first such PRE, as shown in Fig.~\ref{fig:prot1},
that needs only a single classical source and makes no independence assumptions.
The source can be known completely to the adversary (and only has entropy to the devices).
It can be arbitrarily correlated with the untrusted devices, with an almost optimal translation of the degree of 
correlation into the quality parameter for the output~\cite{translate}.
Our extractor,  for the first time, circumvents any form of input-structural or \emph{independence} assumptions 
underlying all existing solutions\cite{all-need-independence}.
In conjunction with~\cite{MS}, our extractor is able to extract \emph{arbitrarily long} randomness from untrusted devices using \emph{any} weak source with \emph{constant bits min-entropy with respect to the devices}.
It is also robust against a constant level of device imprecision, a critical property for practical implementations.
Given enough number of devices, the output error of our protocol can be made close to the minimum.
Given a desirable output error $\epsilon$, the number of devices can be made a polynomial in $1/\epsilon$.

The assumptions for our extractor to work form a minimal set in the following sense~\cite{refine}.
First, a single min-entropy source alone (i.e. without any additional resource) is insufficient due to the impossibility of deterministic extraction~\cite{SanthaV84}. Untrusted devices alone (i.e., without any min-entropy source) are not sufficient either,
because the devices can then pre-program their deterministic answers without generating any randomness. Without
any communication restriction between the adversary and the devices, our task would 
become impossible trivially. Such a restriction between the computational components of an extractor and the adversary is also implicitly assumed for classical extractors. If the devices can communicate freely,
there would be effectively a single device.
Then this single device's optimal strategy for minimizing the abort probability can be made deterministic~\cite{single_device}.
Together with the extractor's deterministic
algorithm, we would then have a deterministic extractor, which cannot even extract a single bit from a general min-entropy source.
We note that it would be useful for practical considerations to relax the no-communication restriction. On the other hand,
results assuming no-communication can be useful for those settings as well (e.g. the output min-entropy reduces by the amount
among the devices and the adversary.)
While in principle, there may be other incomparable minimal set of assumptions allowing for randomness extraction,
we successfully remove assumptions required by all current methods: structural restrictions on the input, some forms of independence and the trust on the inner-working
of the quantum device(s).

Our construction needs two existing ingredients: A {\em quantum-proof classical randomness extractor} $\Ext$ and
a seeded PRE $\Picert$.  Intuitively, $\Ext$ uses a globally uniform seed to transform a min-entropy source into an output that is close-to-uniform with respect to an all-powerful quantum adversary (see Section~\ref{sec:result} for a formal definition.)
Our main theorem provides a ``Master Protocol'' for constructing physical randomness extractors from 
{\em any} pair of $\Ext$ and $\Picert$ with matching input/output length. Our main technical contribution is a general principle for proving security when composing multiple untrusted-device protocols.

We introduce a few technical concepts in order to state our main theorem concretely. The {\em error of $\Ext$} is the worst-case,
over all $(n, k)$ sources, standard distance (trace distance) of the input-output-adversary state to the ideal state, where the output is uniformly distributed with no correlation with the input-adversary subsystem. The {\em noise} in an untrusted device describe its deviation from
the performance of an ``ideal'' device. We define noise fairly generally so that our result is applicable in a wide range of settings.
One specific example is when performing  a Bell-test, such as in the well-known CHSH game~\cite{CHSH}, the noise can be defined
to be the gap between the device's success probability with that of the optimal quantum success probability.
The {\em error of an untrusted-device protocol} is the maximum
of two types of errors: the {\em completeness} error and the {\em soundness} error. The completeness error under a fixed level of noise
is the probability of the protocol rejecting an implementation where the device(s) used are within the specified noise level to
the ideal device(s). The soundness error quantifies the chance of accepting an undesirable output. 

\begin{theorem}[Main Theorem (Informal)]\label{thm:main:informal}
 Let $(\Ext, \Picert)$ be a pair of quantum-proof classical randomness extractor and seeded PRE such that
the output length of $\Ext$ is the same as the input length of $\Picert$. Suppose that $\epsilon$ upper-bounds both the errors of 
$\Ext$ on any $(n, k)$ source and $\Picert$ for a certain noise level. Then the composition of multiple instances of $\Ext$ and $\Picert$ shown in Fig.~\ref{fig:prot1}
with $\eta=\sqrt{\epsilon}$ is a seedless PRE whose error for the same noise level and on an $(n, k)$ source is $O(\sqrt{\epsilon})$.
Furthermore, the source can be known to the adversary and the min-entropy required is with respect to the devices only.
\end{theorem}

The last property of allowing the source to be public means that the randomness is extracted from the untrusted devices, and may have significant practical implication. This is because it allows one to
use a reputable public service, such as NIST's Randomness Beacon~\cite{NIST},
when the users are sufficiently confident that the device makers have little knowledge of the public randomness.

Different choices of $\Ext$ and $\Picert$ give different instantiations of the Master Protocol with different advantages.
We highlight the following three instantiations. 
(1) \emph{The weakest sources.} Using any $(n, k)$ sources of (sufficiently large) constants $n$ and $k$, we can achieve
a constant extraction rate with a constant error for an unbounded output length.
(2) \emph{Minimizing error.} Given sufficiently many devices, our method can reach an error $2^{-\Omega(k^\mu)}$, where $\mu\ge 1/2$ is a universal constant. 
(3) \emph{High min-entropy sources.} For a polynomial entropy rate (i.e., $k \geq n^{\alpha}$ for $\alpha \in (0,1)$), we can extract from $\poly(n)$ untrusted devices with an inverse polynomial error (i.e., $k^{-\beta}$ for $\beta \in (0,1)$) in $\poly(n)$ time.
The Miller-Shi expansion protocol~\cite{MS} is the strongest known $\Picert$ in many aspects thus is used
to achieve robustness and unbounded extraction.
For $\Ext$, (1) uses (repeatedly) a one-bit extractor~\cite{KonigT:extractor}, (2,3) use Trevisan's extractors~\cite{Trevisan:extractor, DPVR12}
(as in Corollary 5.3 and 5.6 of De {\em et al.}~\cite{DPVR12}, respectively).

We sketch the proof for the Main Theorem here. A foundation for all known untrusted-device protocols is to test the super-classical behavior of the devices using the classical source. 
The main challenge for our seedless extraction is to perform such a test with only a given amount of min-entropy to the devices, without any structural or independence assumptions.
Our solution is in essence a reduction of seedless extraction to the syntactically easier task of seeded extraction.
We first improve the input randomness {\em locally}. By the property of the quantum-proof strong extractor $\Ext$,
the output $S_{00...0}\cdots S_{10...0}\cdots S_{11...1}$ of the $\Ext$ instances forms a ``quantum somewhere randomness (QSR)'' source,
in that most of the blocks $S_i$ are almost uniform to the devices. Call such a block ``good.''
Next, each good $S_i$ is transformed by the corresponding $\PExt$ to be near uniform to the adversary.
This transformation {\em decouples} the correlation between a good $S_i$ with the rest of the blocks,
ensuring the near-perfect randomness of the final output.

\paragraph{Equivalence Lemma: a principle for proving composition security.} 
Note that for ``decoupling'' to be meaningful, the source is in general only (close to) uniform-to-device but may be arbitrarily correlated otherwise.
Thus the decoupling feature of the seeded extractors
does not follow directly from their definition or the original proof for their security~\cite{VV12, MS}, which
require a globally uniform input. Our main technical contribution is the following ``Equivalence Lemma'' that bridges the gap
in the input requirements in full generality.

\begin{lemma}[Equivalence Lemma (informal)]\label{lm:eq:informal}
The performance of a seeded physical randomness extractor remains the same
when its uniform-to-all input is replaced by a uniform-to-device input.
\end{lemma}

As a basic principle for securely composing untrusted-device protocols, Equivalence Lemma has found other applications.
We describe two most striking applications (besides the main result).

The first is on {\em unbounded randomness expansion}, i.e., seeded extraction where then output length does not depend on
the input length. Whether or not one could expand randomness securely beyond the exponential rate first shown by Vazrani and 
Vidick~\cite{VV12} was a natural question~\cite{unbounded}.
Intuitively, unbounded expansion is possible because the untrusted devices are randomness-generating.
Indeed, for any $N$, repeating an expansion protocol $O(\log^*N)$ times using a different set of devices
each time expands a seed of a fixed length $N$ output bits. A folklore method for achieving unbounded expansion using
a constant number of devices is to cross-feed two expansion protocols, i.e., using the output of one as the input to the other.
Through an intricate analysis, Coudron and Yuen~\cite{CY} showed that a specific cross-feeding protocol is indeed secure.

The Equivalence Lemma immediately implies that the cross-feeding protocol using {\em any} two expansion protocols is
secure. This is simply because for each expansion protocol,
the input is always (almost) uniformly random to its devices (thus the output is always (almost) uniformly random to the next set of devices.)
This in particular implies that using the robust expansion protocol of Miller and Shi~\cite{MS} gives a {\em robust} unbounded
expansion protocol. The protocol analyzed in~\cite{CY} requires that an honest implementation must tend to an ideal implementation
as the output length grows (thus if the output length is chosen after the device is given, either the device has to be perfect or the output length
cannot be unbounded.)~\footnote{An earlier version of this work~\cite{CSW:QIP} containing the Equivalence Lemma and~\cite{CY} were independent, though we did not state this application there.}

The second significant implication is that {\em public} randomness an be used to produce private randomness,  as long
as the public randomness is uniform to the untrusted devices. This implication holds for both random number generation and
key distribution. A specific scenario that this implication can be of significant practical value is the following.
The NIST Randomness Beacon project~\cite{NIST}
aims to broadcast true randomness to the public. Since the bits become known after broadcast, one cannot use them
directly for cryptographic applications. However, as long as one is willing to assume that the public randomness
is uniform with respect to the untrusted-devices, it can be used securely to generate private randomness.
A related yet subtly different application is that in adapting the Miller-Shi randomness expansion protocol~\cite{MS} 
for key distribution, the Lemma allows the use of locally generated uniform randomness as the initial seed, despite
the original expansion protocol requiring  global randomness.

\paragraph{Physics Implications.}
\commentout{
Our result implies that unless the world is deterministic, we can in principle create arbitrarily many events and be confident
that their joint distribution is close to uniform. This rules out a ``weak randomness world,'' where
randomness exists in Nature but not in a large and close-to-uniform scale~\cite{smaller_k}.
The previous such dichotomy statements~\cite{ColbeckR:2012, Gallego:2012} 
model weak randomness in Nature by the highly structured and highly random SV-sources~\cite{SV-source},
together with a structural or conditional independence assumption. 
We remove those  assumptions. Our dichotomy statement is asymptotically optimal
in the sense that it requires only a constant, as opposed to a linear, amount of uncertainty for certifying
unbounded output randomness.
}

Our result  provides an approach for mitigating the ``freedom-of-choice'' loophole in Bell test experiments
for refuting hidden local variable theories. Such experiments require
the choice of the measurement settings to be nearly uniformly distributed. By using the output of our protocol,
those experiments remain sound even when
only extremely weak source of randomness is available. We can thus consider the composition of the protocol
with the subsequent Bell tests as a combined test for refuting local hidden variable theory
that, unlike the standard Bell tests, needs only a weak random source for choosing the experiment settings.

This version of our paper focuses on quantum-security. We assume the completeness of quantum mechanics, which means that
the adversary cannot obtain any information other through quantum operations.
A planned update of this work will address non-signaling security.


\section{Preliminaries} \label{sec:prelim}
We assume familiarity with the standard concepts from quantum
information and summarize our notation as follows.

\begin{trivlist}

\item \textbf{Quantum States.} We only consider finite dimensional Hilbert spaces as quantum states in infinite dimensions can be 
truncated to be within a finite dimensional space with an arbitrarily small error. The state space $\A$ of  $m$-qubit is the complex Euclidean space $\complex^{2^m}$. An $m$-qubit quantum state is represented by a density operator $\rho$, i.e., a positive semidefinite operator over $\A$ with trace $1$. The set of all quantum states in $\A$ is denoted by $\density{\A}$. 

The Hilbert-Schmidt inner product on $\lin{\A}$, the operator space of $\A$, is defined by $\ip{X}{Y}=\tr (X^*Y)$,  for all $X,Y \in \lin{\A}$, where $*$ is the adjoint operator. Let $\I_\X$ denote the identity operator over $\X$, which might be omitted
from the subscript if it is clear in the context. 
An operator $U \in \lin{\X}$ is a unitary if $UU^*=U^*U=\I_\X$. The set unitary operations over $\X$ is denoted by $U(\X)$.

For a multi-partite state, e.g. $\rho_{ABE} \in \density{\A \ot \B \ot \E}$, its reduced state on some subsystem(s) is represented by the same state with the corresponding subscript(s). For example, the reduced state on $\A$ system of $\rho_{ABE}$
is $\rho_A=\tr_{\B\E}(\rho_{ABE})$, and $\rho_{AB}=\tr_{\E}(\rho_{ABE})$.  When all subscript letters are omitted, the notation
represents the original state (e.g., $\rho=\rho_{ABE}$).

A classical-quantum-, or cq-state $\rho \in \density{\A \ot \B}$ indicates that the $\A$ subsystem is classical and $\B$ is quantum. Likewise for ccq-, etc., states.

We use $\ket{\psi}$ to denote the density operator (i.e., $\ketbra{\psi}$) for a pure state $\ket{\psi}$ when it is clear from the context. Use $\uniform{A}$ to denote the completely mixed state on a space $\A$, i.e., $\uniform{A}=\frac{1}{\dim(\A)} \I_\A$.

\item \textbf{Norms.} For any $X \in \lin{\A}$ with singular values $\sigma_1,\cdots, \sigma_d$, where $d=\dim(\A)$, the trace norm of $\A$ is $\trnorm{X}=\sum_{i=1}^d \sigma_i$.
The trace distance between two quantum states $\rho_0$ and $\rho_1$ is $\trnorm{\rho_0-\rho_1}$. 
Their \emph{ fidelity},  denoted by $\fid(\rho_0, \rho_1)$), is
\begin{equation} \label{prelim:eqn:fidelity}
  \fid(\rho_0, \rho_1)=\trnorm{\sqrt{\rho_0}\sqrt{\rho_1}}.
\end{equation}

The trace distance and the fidelity satisfy the following relations.

\begin{lemma}[Fuchs-van de Graaf] \label{prelim:lem:van_de_Graaf}
For any $\rho_0,\rho_1 \in \density{\A}$, we have
\begin{equation} \label{chap_math:def:eqn_van_de_Graaf}
 1-\frac{1}{2}\trnorm{\rho_0-\rho_1} \leq \fid(\rho_0,\rho_1) \leq
 \sqrt{1-\frac{1}{4}\trnorm{\rho_0-\rho_1}^2}.
\end{equation}
\end{lemma}
The fidelity between subsystems of quantum states cancan be preserved in the following sense.
\begin{lemma}[Folklore]
\label{prelim:lem:fidelity}
Let $\rho,\xi\in \density{\A}$ and $\rho'\in \density{\A \ot \B}$ be density operators with $\tr_{\B}{\rho'}=\rho$.
There exists a density operator $\xi'\in\density{\A \ot \B} $ with $ \tr_{\B}{\xi'}=\xi$ and $\fid(\rho',\xi')=\fid(\rho,\xi)$.
\end{lemma}

\item \textbf{Quantum Operations}. Let $\X$ and $\Y$ be state spaces.
A {\em super-operator} from $\X$ to $\Y$ is a linear map
\begin{equation}
  \Psi : \lin{\X} \rightarrow \lin{\Y}.
\end{equation}
Physically realizable \emph{quantum operations} are represented by \emph{admissible} super-operators,
which are completely positive and trace-preserving.
Thus any quantum protocol can be viewed as an admissible super-operator. We shall use this abstraction in our analysis and make use of the following observation.

\begin{fact}[Monotonicity of trace distances] \label{prelim:fact:monotone_trace}
For any admissible super-operator $\Psi: \lin{\X}\rightarrow \lin{\Y}$ and $\rho_0,\rho_1\in \density{\X}$, we have
\begin{equation}
 \trnorm{\Psi(\rho_0)-\Psi(\rho_1)}\leq \trnorm{\rho_0-\rho_1}.
\end{equation}
\end{fact}

A unitary operation $U \in U(\X)$ is a special type of admissible quantum operations that are \emph{invertible}. For any unitary $U$, its corresponding super-operator $\Psi_U$ is defined as 
\begin{equation}
   \Psi_U(\cdot)= U \cdot U^\dagger.
\end{equation}

Let $\{ \ket{i} : 1\le i\le \dim(\X)\}$ be the computational basis for $\X$.
An $\X$-controlled unitary on $\Y$ is a unitary 
 $U \in U(\X \ot \Y)$ such that 
for some  $U_{i}\in U(\Y)$, $1\le i\le \dim(\X)$,
\begin{equation}
  U =\sum_{1\le i\le\dim(\X)} \ketbra{i} \ot U_{x_i}.
\end{equation}
Likewise define an $\X$-{\em controlled} admissible quantum operation $T$
from $\K$ to $\L$ as an admissible quantum operation such that for some admission
quantum operations $T_i:L(\K)\to L(\L)$, $1\le i\le\dim(\X)$
\begin{equation} 
T = \sum_{1\le i\le \dim(\X)} \langle i|\cdot |i\rangle|i\rangle\langle i|\otimes T_i(\cdot).
\end{equation}


\item \textbf{Min-entropy}.
For a cq state $\rho_{XE}$, the amount of \emph{extractable} randomness (from $X$ against $E$) is characterized by its \emph{(smooth) conditional min-entropy}.
\begin{definition}[conditional min-entropy] \label{SW:def:min-entropy}
Let $\rho_{XE} \in \density{\X\ot \E}$. The \emph{min-entropy} of $X$ conditioned on $E$ is defined as
  \begin{equation*}
    \Hmin({X|E})_\rho \defeq \max \{\lambda
    \in \mathbb{R} :  \exists \sigma_E \in \density{\E}, \mathrm{s.t.}\,\, 2^{-\lambda} \I_X \ot \sigma_E \geq \rho_{XE}\}.
  \end{equation*}
\end{definition}


\end{trivlist}


\section{Formal definitions of Physical Randomness Extractors} \label{sec:formal_model}
We now proceed with formal definitions. We first formalize the notion of physical systems.
By an input or output, we mean a finite length binary string.

A {\em quantum device} $D$ is a Hilbert space, also denoted by $D$,
together with an admissible quantum operation, called its {\em device operation}, which takes a 
classical input, conditioned on which applies a quantum operation on $D$, then produces a classical output.
A \emph{physical system} $\QPS=(\Abs, \vecPhy, \Adv)$ consists of
three disjoint subsystems:
a source  $\Abs$, which is always classical, $t$ quantum devices $\vecPhy=(\Phy_1,\cdots, \Phy_t)$, 
for some $t\ge0$, and a quantum adversary $\Adv$.
We write $\QPS=\QPS(\rho, \{\Alg_{\Phy_i}\})$ to denote that the device operations are $\{\Alg_{\Phy_i}\}$
and the system state is currently $\rho$. Likewise for writing $\QPS=\QPS(\rho)$.
Note that the assumption of no-communication among the devices is formally captured by that each device algorithm $\Alg_{\Phy_i}$ operates only on its corresponding space $\Phy_i$.

As randomness is relative, we will say that in a multi-partite state, a certain classical component has a certain min-entropy {\em (with respect) to another component}.
Similarly we add a scope of subsystems to which a certain classical component is an $(n, k)$-source  or uniformly distributed.
If the scope is the rest of the system, we refer to it as ``global.''
With those conventions, we quantify the min-entropy of a physical system below.

A physical system $\QPS(\rho, \{\Alg_{\Phy_i} \})$ is an $(n,k,t,m)$-physical system with a random-to-devices source if
(1) $\Abs$ is an $(n,k)$-source to the devices, and (2) Each device $\Phy_i$ can only output at most $m$ bits in total or any additional bit of output will not be used.

By replacing ``random-to-devices'' with {\em globally random}, {\em uniform-to-devices}, and {\em globally uniform}, we
define  the corresponding physical systems similarly. For the latter two cases, we omit the min-entropy 
to call $\QPS$ an $(n, t, m)$-physical system.
We now define the syntax of physical randomness extractors.
\begin{definition}[Physical Randomness Extractor]  \label{def:Phy_Ext}
A physical randomness extractor $\PExt$ for a physical system $\QPS(\Abs, \vecPhy, \Adv)$  is a classical deterministic algorithm that
conditioned on $X$, classically interacts with the devices by invoking the device operations, and finally outputs a decision bit $A\in \{0, 1\}$, where $0$ is for rejecting and $1$ for accepting,
and an output string $Z \in \zo^*$ to the corresponding registers ${\DecisionSystem} \ot {\Zsystem}$.
(See Fig.~\ref{fig:phy_ext}.)

The {\em extraction operation} 
\begin{equation}
 \Phi_{\PExt} : \mathrm{L}(\Abs \ot \vecPhy) \rightarrow \mathrm{L}({\DecisionSystem} \ot {\Zsystem} \ot \Abs \ot \vecPhy)
\end{equation}  
is the $X$-controlled admissible operation from $\vecPhy$ to ${\DecisionSystem}\ot{\Zsystem}\ot\vecPhy$ induced by the composition of $\PExt$
and the device operations. 
\end{definition}

When discussing post-extraction states, it will be convenient to say that $\QPS$ is equipped with the registers $\DecisionSystem \ot {\Zsystem}$,
and denote the extended physical system by $\QPS_{{\DecisionSystem}{\Zsystem}}$. Denote by $\AcceptEvent(\PExt,\QPS)$ the event that $\PExt$ accepts when applied to $\QPS$. 

In order to discuss the quality of a PRE, we need the following relative notions of (approximate) uniform distribution.
For an $\epsilon\in[0,2]$, we say that $X$ is $\epsilon$-uniform-to-E in a cqq state $\rho_{XEE'}$ if there exists a $\rho'_{XEE'}$ where $X$ is uniform-to-$E$
and $\|\rho_{XEE}-\rho'_{XEE}\|_{\tr}\le \epsilon$.
Let $\QPS_{{\DecisionSystem}{\Zsystem}}$ be a physical system equipped with the decision-output registers ${\DecisionSystem}\ot {\Zsystem}$, and it is
in a state $\gamma=\gamma_{{\DecisionSystem}{\Zsystem}X{\vecPhy}E}$.
Denote by $\gamma^{\AcceptEvent}=\gamma^{\AcceptEvent}_{{\DecisionSystem}{\Zsystem}X{\vecPhy}E}$ the (subnormalized) 
projection of $\gamma$ to the $A=1$ subspace.

\def\Honest{\textrm{Honest}}
\def\All{\textrm{All}}
\def\PS{\mathbb{S}}

We will discuss noise model abstractly, i.e. independent of the technology implementing the devices.
\begin{definition}[Implementation] An {\em implementation} of devices  $\vecPhy=(\Phy_1, \cdots, \Phy_{\ndevice})$ is a device state $\rho_{\vecPhy}$
together with the device operations $\{\Alg_{\Phy_i}\}_{1\le i\le t}$. 
\end{definition}

For a physical system $\QPS$ over a set $D$ of devices, denote by $\Pi(\QPS)$ its implementation. If $D'$ is a subset of $D$,
denote by the restriction of $\Pi$ to $D'$ by $\Pi_{D'}$.

Recall that a {\em premetric} on a set $A$ is a function $\delta: A\times A \to \mathbb{R}$ such that
$\delta(a, a')\ge0$ and $\delta(a,a)=0$, for all $a,a'\in A$.
We require a noise model to be reasonable in that the noise of a larger system is no less than the noise in a smaller system.

\begin{definition}[Noise Model]
A {\em noise model}  is a premetric on implementations that takes values in $[0,1]$ and is non-increasing
under taking device restrictions. 
\end{definition}

More precisely, let $\delta$ be a noise model, $t\ge0$ be an integer, $D$ be a set of devices,  and $\Pi$ and $\Pi'$ be two implementations of $D$. Then (1) $\delta(\Pi, \Pi') \in [0, 1]$, (2) $\delta(\Pi, \Pi')=0$ if $\Pi=\Pi'$, and (3) ({\em Reasonable Property})~\cite{reasonable}.
If $D'$ is a subset of $D$, then
$\delta(\Pi_{D'}, \Pi_{D'}')\le\delta(\Pi, \Pi')$.

To define the soundness error of a PRE, we need to define that of 
a post-extraction state.
Let $\gamma$ be a post-extraction state described above. 
A subnormalized state $\alpha=\alpha_{{\DecisionSystem}{\Zsystem}X{\vecPhy}E}$ is called {\em ideal} if
$A=1$ (i.e. $\alpha=|1\rangle\langle1|_{\DecisionSystem}\otimes
\alpha_{{\Zsystem}X{\vecPhy}E}$) and ${\Zsystem}$ is uniform to $XE$.
We say that $\gamma$ has a soundness error $\epsilon$ if there exists an 
ideal post-extraction state $\alpha$ such that 
$\trnorm{ \gamma^{\AcceptEvent} - \alpha }\le \epsilon$.
We are now ready to define properties of physical randomness extractors.
\begin{definition}[Soundness, Completeness, and Robustness of a PRE] \label{def:PRE_para}
Let $\PS$ be a non-empty set of physical systems, $\delta$ a noise model,
$\eta\in[0,1]$, and $\PExt$ a PRE. We say that $\PExt$ is an {\em untrusted-device} PRE for $\PS$
and has a completeness error $\eps_c$ tolerating an $\eta$ level of noise, 
and a soundness error $\eps_s$, if the following completeness and soundness properties
hold. 
\begin{itemize}
\item (Completeness) There exists an implementation $\Pi^*$, referred to as the {\em ideal} implementation,
in the implementations of $\PS$,
such that for all $\QPS\in\PS$ whose implementation $\Pi$ satisfies $\delta(\Pi, \Pi^*)\le \eta$,
$\Prob[ \AcceptEvent(\PExt,\QPS)] \geq 1- \eps_c$.
\item (Soundness) For any $\QPS(\rho)\in\PS$, $\Phi_{\PExt}(\rho)$ has a soundness error $\le \eps_s$.
\end{itemize}

We further call $\PExt$ a {\em random-to-devices $(n, k, t, m)$-PRE,} 
for integers $n,k,t,m\ge0$, if $\PS$ is the set of all $(n, k, t, m)$-physical sources with a 
random-to-devices source. Likewise define the notions of a {\em random-to-all $(n, k, t, m)$-PRE},  
a {\em seeded} $(n, t, m)$-\PRE with a {\em uniform-to-devices} seed, and a {\em seeded} $(n, t, m)$-\PRE
with a {\em uniform-to-all} seed. 
If $N$ is the (maximum) output length of $\PExt$, the {\em (extraction) rate} of $\PExt$ is $N/(mt)$.
\end{definition}

\noindent Note that our soundness definition requires the output to be uniform with respect to both the source $\Abs$ and the adversary $\Adv$, which implies that the randomness $\PExt$ extracts is from the devices $\vecPhy$. 

Previous works~(e.g. \cite{VV12}) define $\epsilon$ to be a soundness error if either the protocol accepts with $\le\epsilon$ probability or
the state conditioned on accepting has the desired amount of randomness. 
While our definition is essentially equivalent, syntactically it has several new and subtle features that greatly
simplify the analysis of PRE compositions. First, a single inequality for the definition avoids the otherwise necessary
argument about conditional property (conditioned on accepting). Second,
use the whole state, as opposed to tracing out the device component, in calculating the distance to an ideal state (that has a uniform $X$).
Third, the ideal state in comparison does not need to have the same Adversary subsystem. Those features allow
the application of triangle inequality to  the case of perturbed input state and consequently, composed protocols.

\vspace{1mm} \noindent \textbf{Previous randomness expansion protocols seen as seeded-PREs.} 
By definitions, randomness expansion protocols~\cite{col:2006,ColbeckK:2011,VV12,MS} are precisely
seeded PREs with uniform-to-all seeds. There has been a large body of research on randomness expansion protocols since~\cite{col:2006}.
Our framework allows deeper quantitative analyses and comparisons of their performances.

Phrased in our framework,  Vazirani-Vidick~\cite{VV12} showed that a quantum-secure $2$-device PRE needs only
a poly-logarithmic seed length (measured against the output length) and can achieve an inverse polynomial extraction rate
and an inverse polynomial error (in the output length).
The concurrent work of Miller-Shi~\cite{MS} significantly improved the rate to be {\em linear}, and the error to be
negligible (inverse quasi-polynomial), besides adding the constant-noise robustness feature.
Another concurrent work of Coudron and Yuen~\cite{CY} reduces the seed length to a \emph{constant}
at an inverse polynomial rate using $8$ devices.Finally, combining our technique (Equivalence Lemma in the next section) with Miller-Shi~\cite{MS} in a straightforward matter, we can achieve simultaneously a linear rate, a constant noise robustness and a constant seed length.

\section{Results}\label{sec:result}
In this section, we will present the precise statements of the Equivalence Lemma and the Main Theorem with
the explicit construction of our main protocol and the necessary tools for completing our analysis. We leave
all the proofs in the Appendix.

Our analysis uses the seeded PRE as a black-box. This differs from previous analyses that
rely on the details of the untrusted-device protocols (e.g., those for randomness expansion such as~\cite{VV12,MS} or for other tasks, such as delegation of quantum computation~\cite{ruv:2013} and certifying strong monogamy~\cite{Mas09, CSW:QIP}). 

We fix a noise model, and refer to the triple $(\epsilon_c, \eta, \epsilon_s)$ 
of a completeness error $\epsilon_c$, a noise level tolerated $\eta$,
and a soundness error $\epsilon_s$ as the {\em performance parameters}.
We shall prove the lemma below. While the proof is short, the result may appear surprising or even counter intuitive.

\begin{lemma}[Equivalence Lemma]  \label{lem:EL}
Any seeded PRE for uniform-to-all seeds is also a seeded PRE for uniform-to-devices seeds with the same performance
parameters under the same noise model.
\end{lemma}

Our proof consists of two steps, the first of which is general and the second is specific for the setup of the Lemma.
For a set of system states $\PS$, denote by $\PS'$ the set of system states that can be obtained from a $\QPS(\rho_{XDE})\in\PS$
by applying a $X$-controlled operation on $E$.

\begin{proposition} Let $\PS$ be a set of physical systems and $\PExt$ a PRE.
Then $\PExt$ has the same performance parameters on $\PS$ and $\PS'$.
\end{proposition}

\begin{proof}Since $\PS\subseteq\PS'$, we need only to show that the performance parameters of $\PS'$ are
no worse than those of $\PS$.  Fix a $\QPS'(\rho'_{XDE'})\in\PS'$.
Let ${M}$ be an $X$-controlled $E$-to-$E'$ operation and $\QPS(\rho_{XDE})\in\PS$ be such that $\rho'={M}(\rho)$.
 
{\em Completeness and robustness.} Use the same ideal implementation $\Pi$ for $\PS$ as the ideal implement for $\PS'$.
Assume that $\delta(\Pi(\QPS'), \Pi)\le \eta$.
Since $\Pi(\QPS)=\Pi(\QPS')$, we have $\delta(\Pi(\QPS), \Pi) \le \eta$. Thus
$\AcceptEvent(\QPS)\ge 1-\epsilon_c$. Since the acceptance probability depends only on the reduced density operator on $XD$,
and $\rho'_{XD} = \rho_{XD}$. We have ${\AcceptEvent}(\QPS')\ge 1-\epsilon_c$. This proves the claimed result on completeness
and robustness.

{\em Soundness.} 
Note that $\Phi_{\PExt}$ commute with ${M}$ as both are $X$-controlled operations acting on two
disjoint quantum subsystems. That is,
\begin{equation} \Phi_{\PExt}(\rho') = {M} (\Phi_{\PExt}(\rho)).\end{equation}
Furthermore,
\begin{equation} \Phi_{\PExt}(\rho')^\AcceptEvent = {M} \left(\Phi_{\PExt}(\rho)^\AcceptEvent\right).\end{equation}

By the soundness of $\PExt$ on $\PS$, for an ideal post-extraction state $\delta$,
\begin{equation}\left \| \Phi_{\PExt}(\rho)^{\AcceptEvent} - \delta\right\|_{\tr} \le \epsilon_s.\end{equation}
Thus
\begin{equation}
\left \| \Phi_{\PExt}(\rho')^\AcceptEvent - M(\delta)\right\|_{\tr} =
\left \| {M}\left(\Phi_{\PExt}(\rho)^{\AcceptEvent}\right) - M(\delta)\right\|_{\tr} 
\le \epsilon_s.\end{equation}
Since $M(\delta)$ also has $A=1$ and $Y$ uniform to $XE$,  $\Phi_{\PExt}(\rho')$ has also an $\epsilon_s$ soundness error. This proves the soundness claim.
\end{proof}

The Equivalence Lemma follows from the above together with the following proposition.

\begin{proposition} Let $D$ be a set of devices. If $\PS_D$ is the set of all global-uniform physical systems
over $D$, $\PS'_D$ is the set of all device-uniform physical systems over $D$.
\end{proposition}
\begin{proof} We shall omit the subscript $D$ in this proof.
Clearly all states in $\PS'$ are device-uniform. 
We need only to show that for an arbitrary device uniform state
$\rho'=\rho'_{XDE}$, there exists a global-uniform $\rho=\rho_{XDE'}$ and an $X$-controlled operation $M$ on $E'$ such
that $\rho'={M}(\rho)$.

We write $\rho'= \sum_{1\le x\le\dim(X)} |x\rangle\langle x|\otimes\rho'^x_{DE}$.
Let $E'$ be a system of the same dimension as $DE$, and for each $x$, let $|\phi^x_{DEE'}\rangle$
be a purification of $\rho'^x_{DE}$. That is, with $\phi^x= |\phi^x\rangle\langle\phi^x|$,
$\phi^x_{DE}=\rho'^x_{DE}$. By the assumption that $\rho'_{XD}= U_X\otimes\rho'_D$, 
$\phi^x_D$ are  identical for all $x$. Set $|\phi\rangle=|\phi^{0^n}\rangle$. Thus by Uhlmann's Theorem (c.f. Chapter 9 of~\cite{Nielsen:2000:book}),
for each $x$, there exists a unitary operators $U_x$, such that
\begin{equation} |\phi^x\rangle = U^x|\phi\rangle.\end{equation}
Define ${M}$ to be the $X$-controlled $E$-operation that is the composition of applying $U^x$, as controlled by $X$, then tracing out  $E'$.
Then with $\rho=\sum_{x} |x\rangle\langle x|\otimes \phi\in\PS$,
\begin{equation} \rho' = {M} (\rho).\end{equation}
Therefore, $\rho'\in\PS'$.
\end{proof}

\paragraph{Quantum-proof strong randomness extractors and Somewhere Random Source} \label{sec:somewhere}
We will first review quantum-proof
randomness extractors, which turn a min-entropy source to a quantum-secure output, with the help of a short seed.
Then we will introduce the somewhere random sources in our protocol construction.

\begin{definition}[Quantum-proof Strong Randomness Extractor] \label{SW:def:q_strong_extractor}
  \label{def:extractor_q_proof}
  A function $\Ext: \{0,1\}^n \times \{0,1\}^d \to \{0,1\}^m$ is a
  \emph{quantum-proof} (or simply \emph{quantum})
  \emph{$(k,\eps)$-strong randomness extractor}, if for all cq
  states $\rho_{XE}$ with $\Hmin(X|E) \geq k$,
  and for a uniform seed $Y$ independent of $\rho_{XE}$, we
  have
\begin{equation} \label{eqn:extractor_q_proof}  \trnorm{
    \rho_{\Ext(X,Y)YE} - \uniform{m} \ot \rho_Y \ot \rho_E}
  \leq \eps. \end{equation}
 \end{definition}
 
It is known that Trevisan's extractors~\cite{Trevisan:extractor} are secure against quantum adversaries~\cite{DPVR12}.
Those will be used for instantiating our main theorem.
An apparent problem when one tries to apply those extractors in our setting is that we do not have
the required uniform seed. Our solution is to enumerate all the possible seed values and run the extractors
on the fixed seed values. The output property of the extractor now translates to a guarantee that the output of
at least one instance (in fact, a large fraction of them) of the fixed-seeded extractors is close to uniform. The output together forms what we call
{\em quantum somewhere randomness}.
 In classical setting, a somewhere random source $\vecS$ is simply a sequence of random variables $\vecS = (S_1,\dots,S_r)$  such that the marginal distribution of some block $S_i$ is uniform (but there can be arbitrary correlation among them). Somewhere random sources are useful intermediate objects for several constructions of randomness extractors (see, e.g., \cite{Rao07,Li13}), but to the best of our knowledge, its quantum analogue has not been considered before. 
 
\begin{definition}[Quantum-SR Source] \label{SW:def:SR_q_source} A cq-state $\rho \in \density{S_1\ot \cdots \ot S_r \ot E}$ with  classical  $S_1$, $S_2$, $\cdots$, $S_r \in \01^m$ and quantum  $E$ is a \emph{$(r,m)$-quantum somewhere random (SR) source} against $E$ if
there exists $i \in [t]$ such that 
\begin{equation}
  \rho_{S_iE} = \uniform{m} \ot \rho_E.
\end{equation}
We say that $\rho$ is a \emph{$(r,m,\eps)$-quantum somewhere random source} if there exists $i \in [t]$ such that
  \begin{equation}
      \trnorm{\rho_{S_iE}-\uniform{m}\ot \rho_E}\leq \eps.
  \end{equation}
%
\end{definition}

We remark that the fact that $\rho$ is a $(r,m,\eps)$-quantum somewhere random source does not necessarily imply that $\rho$ is $\eps$-close in trace distance to some $(r,m)$-quantum somewhere random source $\rho'$. 
In contrast, the analogous statement is true for classical somewhere random source. However, by Lemma~\ref{prelim:lem:fidelity}, one can show that they are $2\sqrt{\eps}$ close.  On the other hand, 
just like its classical counterpart, one can convert a weak source $X$ to a somewhere random source by applying a (quantum-proof) strong randomness extractor to $X$ with all possible seeds (Each seed yields one block). 

\begin{proposition} \label{prop:SR_generation}
Let $\Ext: \{0,1\}^n \times \{0,1\}^d \to \{0,1\}^m$ be a
  \emph{quantum-proof} \emph{$(k,\eps)$-strong extractor}. Let $\rho_{XE}$ be a cq-state with $\Hmin(X|E) \geq k$.
  For every $i\in \zo^n$, let $S_i = \Ext(X,i)$. Then the cq-state
  \begin{equation}
    \rho_{S_1 \dots S_{2^d} E} \defeq \sum_x p_x \ketbra{S_1} \ot \cdots \ot \ketbra{S_{2^d}} \ot \rho_E^x,
  \end{equation}
  is a $(2^d,m, \eps)$-quantum SR source. 
  Moreover, the expectation of $\trnorm{\rho_{S_iE}-\uniform{m}\ot \rho_E}$ over a uniform random index $i\in \zo^n$ is at most $\eps$. 
\end{proposition}

\begin{proofof}[Proposition~\ref{prop:SR_generation}]
Since $\Ext$ is a quantum-proof $(k,\eps)$-strong extractor and $\Hmin(X|E) \geq k$, we have that
\begin{equation}
  \trnorm{\rho_{\Ext(X,Y)YE} - \uniform{m} \ot \rho_Y \ot \rho_E} \leq \eps,
\end{equation}
which is equivalent to
\begin{equation}
  \sum_{i=1}^{2^d} \frac{1}{2^d} \trnorm{\rho_{\Ext(X,i)E}-\uniform{m} \ot \rho_E} \leq \eps.
\end{equation}
Thus immediately we have that there exists an index $i \in [2^d]$ such that
\begin{equation}
  \trnorm{\rho_{\Ext(X,s_i)E}-\uniform{m} \ot \rho_E} \leq \eps,
\end{equation}
or equivalently $\trnorm{\rho_{S_iE}-\uniform{m} \ot \rho_E}\leq \eps$.
%
\end{proofof}

\paragraph{Construction of PREs for any min-entropy source} \label{sec:protocol}
We are now able to state precisely our main theorem. The Master Protocol is described in details in Fig.~\ref{fig:prot_A}.

\begin{figure}
\begin{protocol*}{Physical Randomness Extractor~$\PExt$}
\begin{enumerate}[itemsep=-3pt]
\item Let $\Ext: \zo^n \times \zo^d \rightarrow \zo^m$ be a quantum-proof strong randomness extractor.
\item Let $\PExt_{\mathrm{seed}}$ be a seeded PRE with seed length $m$ that uses $\tcert$ devices. Let $0<\eta<1$. 
\item $\Piamp$ operates on an input source $X$ over $\zo^n$ and $\tamp = 2^d \cdot \tcert$ devices $\vecPhy = (\vecD_1,\dots,\vecD_{2^d})$, where each $\vecD_i$ denotes a set of $\tcert$ devices, as follows.
\end{enumerate}
\begin{step}
\item For every $i \in \zo^d$, let $S_i = \Ext(X,i)$ and invoke $({\AcceptEvent} _i, Z_i) \leftarrow \Picert(S_i, \vecD_i)$.
\item If there exist $\eta$ fraction of ${\AcceptEvent} _i = 0$, then $\Piamp$ outputs $A=0$; otherwise, $\Piamp$ outputs $(A,Z) = (1, \bigoplus_{i\in [2^d]} Z_i)$.
\end{step}
\end{protocol*}
\caption{\small Our Main Construction of Physical Randomness Extractor $\Piamp$.}
\label{fig:prot_A}
\end{figure}

\begin{theorem}[Main Theorem] \label{thm:main}  Let $\Ext: \zo^n \times \zo^d \rightarrow \zo^m$ be a quantum $(k,\eps_{\Ext})$-strong randomness extractor, $\Picert$ be a 
$(m, \tcert, l)$-seeded-PRE with a uniform-to-all seed of length $m$, a completeness error $\eps_c$ tolerating an $\eta$ level of noise, 
and a soundness error $\eps_s$. Let $0<\eta<1$ be the rejection threshold.
Then $\Piamp$ (as shown in Fig.~\ref{fig:prot_A}) is a $(n, k, 2^d \tcert, l)$-PRE for random-to-device sources with a completeness error $(\eps_c+\eps_{\Ext})/\eta$ and 
a soundness error $\eps_s+2\sqrt{\eps_{\Ext}}+\eta$ .
\end{theorem}

\begin{proof}
We first introduce some notations. We will use $i$, $i=1,\cdots, 2^d$ as a subscript to index an instance of 
$\Ext$ or $\Picert$ corresponding to the seed $i$ for $\Ext$.
Fix a physical system $\QPS(\rho_{XDE})$ with $\Hmin({X|D})\ge k$. We continue to write $\rho_{XSDE}$ to denote the
state after applying the $\Ext$ instances. We write ${\RejectEvent} $ for  $1-\AcceptEvent$ and similarly for ${\RejectEvent} _i$ for each $i$.
For each $i$, define
\begin{equation} w_i=\trnorm{\rho_{S_iDE}-\uniform{S_i}\ot \rho_{DE}}.\end{equation}
By the strong extracting property of $\Ext$, for a uniformly chosen $i$,
\[ E_i [ w_i ] \le \eps_{\Ext}.\]
Also, by Propositions~\ref{prelim:lem:fidelity} and \ref{prelim:lem:van_de_Graaf}, there exists a state $\rho'^{(i)}_{XSDE}$ with $\rho'_{S_iDE} = \uniform{S_i}\ot\rho_{DE}$ and
\[ \trnorm{\rho-\rho'} \le  2\sqrt{w_i}.\]
Fix such a $\rho'^{(i)}$ for each $i$.

{\em Completeness.} For each $i$, fix an ideal implementation $\Pi_{D_i}^*$ for $\Picert_i$. We will use their tensor product $\Pi^*$ as the ideal
implementation for $\Piamp$. 
Now suppose that $\delta(\Pi(\QPS), \Pi^*)\le\eta$.
By the reasonable property of the noise model, for each $i$,
$\delta(\Pi_{D_i}(\QPS), \Pi^*_i) \le \eta$. Let $\tilde\rho=\tilde\rho_{XSDE}$ be an arbitrary state with $\tilde\rho_{S_iDE}=U_{S_i}\otimes\rho_{DE}$.
Then $\tilde\rho$ is uniform-to-$D_i$ and 
\[ \delta(\Pi_{D_i}(\tilde\rho), \Pi_{D_i}^*)
= \delta(\Pi_{D_i}(\rho), \Pi_{D_i}^*)\le \eta,\]
by the completeness of $\Picert$,
\[ {\AcceptEvent} _i(\tilde\rho) \ge 1-\eps_c.\]
Since $\Phi_{\PExt_i}$ acts on $S_iD_i$ only,
\[ {\RejectEvent} _i(\rho)\le {\RejectEvent}_i(\tilde\rho) + w_i \le  \eps_c + w_i.\]
Thus
\[ \Prob[\RejectEvent] = \Prob[\sum_i {\RejectEvent} _i \ge \eta 2^d] \le (\eps_c+E_i[{w_i}])/\eta \le (\eps_c+{\epsilon_{\Ext}})/\eta.\]

{\em Soundness.} 
The proof is based on the following two observations. The first is that if some $Z_i$ is uniform (to $XE$), then so is $Z$.
It follows that if for each $i$, $Z_i$ is uniform in a (subnormalized) $\gamma_i$, then the $Z$ obtained from $\sum_i\gamma_i$ is also uniform.
Note that not all $\Picert_i$ accepts when $\Piamp$ accepts, thus the observation cannot be directly applied.
This problem is resolved by an additional insight that  for a randomly chosen $i$, the chance of $\Piamp$ accepts but $\Picert_i$ rejects is small.
The details follow.

Denote by  $\vec A=[{\AcceptEvent} _i]_i$ and $\vec Z=[Z_i]_i$. Let $\gamma=\gamma_{X{\vec A}{\vec Z}DE}$ be the state after applying all instances of $\Picert$. Let $\gamma'^{(i)}$ be the same except that $\rho$ is replaced by $\rho'^{(i)}$.
By the soundness of $\Piamp$, there exists an ideal post-extraction (with respect to $\PExt_i$) state $\gamma_i$ such that
\[ \trnorm{\gamma'^{(i) \AcceptEvent_i} - \gamma_i} \le \eps_c.\]
Thus
\[ \trnorm{\gamma^{\AcceptEvent_i} - \gamma_i} \le \trnorm{\gamma^{\AcceptEvent_i} - \gamma'^{(i)\AcceptEvent_i}} +
\trnorm{\gamma'^{(i)\AcceptEvent_i} - \gamma_i}  \le 2\sqrt{w_i} +\eps_c.\]
Applying the final acceptance projection (that is, accept when $<\eta 2^d$ of $A_i$'s reject), we have
\begin{equation}\label{eqn:accept}
\trnorm{\gamma^{\AcceptEvent\wedge\AcceptEvent_i} - \gamma_i^{\AcceptEvent}} \le 2\sqrt{w_i} + \eps_c.
\end{equation}
Note that $\gamma_i^{\AcceptEvent}$ still has $Z_i$ uniform to  $XE$.
Thus with  
\begin{equation}
\gamma' = E_i[\gamma_i^{\AcceptEvent}],
\end{equation}
and $Z(\cdot)$ the super-operator for outputting $Z$,
we have that the subnormalized state $Z(\gamma')$ has its $Z$ uniform to $XE$.
We shall show that $Z(\gamma')$ approximates $\Phi_{\Piamp}(\rho)^{\AcceptEvent}$ well.

\begin{eqnarray}
&&\trnorm{\Phi_{\Piamp}(\rho)^{\AcceptEvent} - Z(\gamma')}\\
&=&\trnorm{\gamma^{\AcceptEvent} - \gamma'}\\
&=&  \trnorm{E_i\left[ \gamma^{\AcceptEvent\wedge\AcceptEvent_i} +\gamma^{\AcceptEvent\wedge\RejectEvent_i}-\gamma_i^{\AcceptEvent}\right]}
\label{eqn:split}\\
&\le& E_i\left[ \trnorm{\gamma^{\AcceptEvent\wedge\AcceptEvent_i} +\gamma^{\AcceptEvent\wedge\RejectEvent_i}-\gamma_i^{\AcceptEvent}}\right]\\
&\le& E_i\left[ \trnorm{ \gamma^{\AcceptEvent\wedge\AcceptEvent_i} - \gamma_i^{\AcceptEvent}}] + E_i[\AcceptEvent\wedge\RejectEvent_i\right]\\
&\le& 2\sqrt{\epsilon_{\Ext}} + \eps_c + \eta.
\end{eqnarray}
Eqn.~(\ref{eqn:split}) is because for any $i$, 
\[\gamma^{\AcceptEvent} = \gamma^{\AcceptEvent\wedge\AcceptEvent_i} + \gamma^{\AcceptEvent\wedge\bar \AcceptEvent_i}.\]
The last inequality is by Eqn.~(\ref{eqn:accept}) and the acceptance criterion. 
Thus we conclude that the soundness error of $\Piamp$ is $\le 2\sqrt{\epsilon_{\Ext}} + \eps_c+\eta$.
\end{proof}

\paragraph{Instantiations.} The Miller-Shi seeded PRE (and its unbounded expansion
composition via ``Equivalence Lemma'') subsumes all other constructions, thus is preferred to use in our instantiations.
Thus the main choice is the quantum-proof classical strong extractors.
We use two known methods for constructing such extractors,
both based on the work of K{\"o}nig and Terhal~\cite{KonigT:extractor} showing
that any classically secure one-bit extractor is automatically secure against quantum adversaries
(with slightly worse parameters.) 
The first method is to take a single-bit extractor and increase the output length by using
independent copies of the seeds. The second is to apply Trevisan's compositions
of the single-bit extractor, which was proved to be quantum-secure by De {\em et al.}~\cite{DPVR12}.

The first instantiation uses the first method by
setting the error parameter of the single-bit extractor (e.g. in Proposition C.5 of~\cite{DPVR12}) 
to be $\Theta(\epsilon/\log^c(1/\epsilon))$, where $c$
is a universal constant from Miller-Shi~\cite{MS}, and
the number of independent seeds to be $O(\log^c(1/\epsilon))$. This requires the min-entropy to be
$O(\log^c 1/\epsilon)$. The number of devices is $(n/\epsilon)^{O(\log^c(1/\epsilon))}$, thus is efficient for constant $\epsilon$.

The second instantiation uses the Trevisan's extractor in Corollary 5.1 of~\cite{DPVR12} for $\Ext$.
Fix an extractor ``seed length index'' $\nu$, defined in Miller-Shi~\cite{MS} as a real so that there exists
a quantum-proof strong extractor on $(n, \Theta(n))$ sources, extracting $\Theta(n)$ bits with an error $\epsilon$
 using a seed length $O(\log^{1/\nu} (n/\epsilon))$. 
Set the error parameter for $\Ext$ to be $\epsilon_{\Ext} = \exp(-k^{\nu})$. It extracts $m=k-o(k)$ bits in the quantum
somewhere randomness output. For $\Picert$, use Miller-Shi's expansion protocol with a constant $q$ parameter
(in their main theorem), an output length $N=k$, and an error $\epsilon_{MS}=2^{-(c k)^{\nu}}$, for a constant $c>0$
to be determined later.
The extractor used inside Miller-Shi is the Trevisan's extractor in Corollary 5.4 of~\cite{DPVR12}, which
requires a seed length of $O(\log^{1/\nu}(N/\epsilon_{MS}))$. Thus the total randomness needed for Miller-Shi
is $O(q\log(1/q)+c)k$, which can be made $\le k$ by choosing sufficiently small $q$ and $c$.
Our Master Protocol now outputs $k$ bits with an error $\exp(-\Omega(k^{\nu}))$.
Applying the unbounded expansion protocol of Miller-Shi on this output, the final error remains $\exp(-\Omega(k^{\nu}))$.
Note that the total number of devices is dominated by the number of $\Ext$ instances, which
is $2^{O(\log^2n +  k^{2\nu}) \log k}$. 

The third instantiation uses the Trevisan's extractor in Corollary 5.6 of~\cite{DPVR12}, with an inverse polynomial
error, and extracting a polynomial fraction of input min-entropy.

\section{Future work}\label{sec:open}
If the error is allowed to be a constant, our construction needs only a source of a sufficiently large constant min-entropy and
length and the output can be arbitrarily long, using just a constant number of devices. However, for much smaller error,
our construction does not achieve simultaneously
close to optimal error parameter and efficiency in the number of devices and the running time. In particular,
the construction cannot reach a cryptographic level of security as the number of devices is at least inverse polynomial
of the error parameter. This raises a fundamental question: {\em is this high complexity necessary?}
A preliminary result of the current authors together with Carl~A.~Miller shows that the number of devices
has to be polynomially related to the input length for any untrusted-device protocol that works on all
weak source of a sufficiently small {\em linear} min-entropy, and the devices are allowed to communicate
in between playing two rounds of non-local games. While this does not answer our question when the devices
are not allowed to communicate throughout the protocol, it indicates the difficulty of reducing the number of devices
in seedless extraction. We do not yet have solid evidence to support a significant reduction in complexity.
A strong lower bound (on the number of devices as function of the error parameter) would have a strong
(negative) interpretation that we will have to resort to some stronger assumptions than those for our theorem
in order to achieve cryptographic level of randomness
generation.

Many other new questions arise from our framework of PRE. Is there an ideal PRE, where all parameters are simultaneously optimize?
Or perhaps there are inherent tradeoffs. Other questions include, what quantities about the untrusted-device determine the maximum amount of output randomness? Can one quantify the restrictions on communication to shed light on its tradeoff with other parameters?
Barrett, Colbeck and Kent~\cite{BCK:memory} pointed out additional potential security pitfalls
in composing untrusted-device protocols. An important direction is to develop a security model in which
one can design PREs and prove composition security in a broad setting.

\section{Acknowledgement}\label{sec:ack}
We thank Carl A. Miller for stimulating discussions. We are also indebted to Roger Colbeck and
Fernando Brand{\~a}o for clarifying their own related works.
\bibliographystyle{abbrv}
\bibliography{../resources/rand_Amp}

\begin{thebibliography}{10}
\bibitem{convex}
This is because the uniform distribution is a convex combination of
  deterministic distributions. Thus any algorithm rejecting all deterministic
  distributions will also reject the uniform distribution.
\bibitem{independence}
Because a convex combination of several independent distributions can be far
  from independent.
\bibitem{quantum_measurements}
According to quantum mechanics, certain simple measurement is guaranteed to
  produce perfect randomness.
\bibitem{SV-source}
A Santha-Vazirani (SV) source has uncertainty in each individual bit even
  conditioned on the value of its previous bits. Namely, a source $X$ over
  $\zo^n$ is $\eps$-SV if for every $i \in [n]$ and $x_{\leq i} \in \zo^{i}$,
  $\eps \leq \Pr[ X_i = x_i | X_{<i} = x_{<i}] \leq 1-\eps$.
\bibitem{general}
The argument made here is on the reliance of physical laws in general. Specific
  theories may be proven not to be the true law of Nature, thus it remains
  possible that a PRE becomes insecure after a scientific revolution, e.g., one
  that disproves quantum mechanics or relativity.
\bibitem{nonsignaling}
Works on those two threads may aim at proving security against an adversary not
  necessarily bounded by quantum mechanics.
\bibitem{translate}
If the source is known completely to the devices then no secure extraction is
  possible and this is reflected in our translation from the correlation to the
  security guarantee. More precisely, if the probability for the devices to
  guess the source correctly is $2^{-k}$, then the error parameter (to be
  defined precisely) has to be at least $2^{-k}$. The error achieved by our
  protocol is $2^{-k^c}$ for some universal constant $c$.
\bibitem{all-need-independence}
Note that trusted quantum devices can be used to generate independent random
  sources readily. Thus, it is a stronger assumption that directly implies
  independence.
\bibitem{refine}
Since all classes of mathematical objects can be further divided into smaller
  classes according to certain features, a minimal set of requirements is
  always relevant to a collection of concepts. Consider for example the
  partition of protocols by whether or not communication among different
  devices are allowed. One can refine the partition by the amount of
  communication, or/and by the pairs of devices that communicate. Thus like all
  claims of minimal assumptions, ours is relative to a certain level of
  partition of protocols.
\bibitem{single_device}
First, the device can simulate any quantum process. Thus the minimum abort
  probability is the value of a classical zero-sum game, for which the device
  has an optimal and deterministic strategy.
\bibitem{unbounded}
One should note that any super-exponential expansion would necessarily make the
  output not suitable for cryptographic applications. This is because when the
  seed length is $k$, the adversary has $\ge2^{-k}$ chance of guessing the seed
  correctly. Thus when the output length is $\exp(\Omega(k))$, this probability
  would be only inverse polynomial of the output length.
\end{thebibliography}

\end{document}